\documentclass [11pt]{article}
\usepackage{amsmath,amsthm,amsfonts,amscd,eucal,latexsym,amssymb,mathrsfs,dsfont,cancel,mathtools,physics,tensor,comment,appendix,enumerate,enumitem,cases,wasysym}
\usepackage{pdfpages}
\pdfoutput=1 
\usepackage[normalem]{ulem}
\usepackage{pst-node}
\usepackage{epsfig,caption,float,xcolor}  
\usepackage{simplewick,bm}
\definecolor{hypercolor}{RGB}{9,80,162}
\usepackage{hyperref,cleveref,cite}
\hypersetup{
	colorlinks = true,
	linkbordercolor = {white},
	linkcolor={hypercolor},
	citecolor={hypercolor},
	urlcolor={hypercolor}
}
\usepackage{csquotes}

\def\cA{{\mathcal A}}

\def\cC{{\mathcal C}}

\def\cF{{\mathcal F}}

\def\cH{{\mathcal H}}

\def\cM{{\mathcal M}}

\def\cR{{\mathcal R}}
\def\cS{{\mathcal S}}

\def\cV{{\mathcal V}}
\def\cW{{\mathcal W}}

\def\beq{\begin{eqnarray}}
\def\eeq{\end{eqnarray}}
\def\pa{\partial}
\def\at{\left(}               		   
\def\aq{\left[}               		   

\def\ct{\right)}              		   
\def\cq{\right]}              		   
\def\d{\mathrm{d}}
\def\e{\mathrm{e}}

\newtheorem{theorem}{Theorem}[section]

\theoremstyle{definition}

\newcommand{\wick}[1]{{:}#1{:}}

\newcommand{\eq}{eq.\@ }
\newcommand{\eqs}{eqs.\@ }
\newcommand{\Eq}{Eq.\@ }

\newcommand{\expvalom}[1]{\expval{\wick{#1}}_\omega}
\newcommand\eqH{\stackrel{\mbox{\tiny $\cH$}}{=}}
\newcommand\leqH{\stackrel{\mbox{\tiny $\cH$}}{<}}
\newcommand\geqH{\stackrel{\mbox{\tiny $\cH$}}{>}}

\newcommand{\enquoteit}[1]{``\textit{#1}''}
\newcommand{\doi}[1]{\href{https://www.doi.org/#1}{[#1]}}

\begin{document}
	
	
\par 
	\bigskip 
	\LARGE 
	\noindent 
	{\bf Evaporation of four-dimensional dynamical black holes sourced by the quantum trace anomaly} 
	\bigskip 
	\par 
	\rm 
	\normalsize


\large
\noindent
{\bf Paolo Meda$^{1,3,a}$}, {\bf Nicola Pinamonti$^{2,3,b}$}, {\bf Simone Roncallo$^{1,3,c}$}, {\bf Nino Zangh\`\i$^{1,3,d}$}  \\
\par

\small
\noindent$^1$ \textit{Dipartimento di Fisica, Universit\`a di Genova, Italy}.

\smallskip

\noindent$^2$ \textit{Dipartimento di Matematica, Universit\`a di Genova, Italy}.

\smallskip

\noindent$^3$ \textit{Istituto Nazionale di Fisica Nucleare - Sezione di Genova, Italy}.

\smallskip

\newcommand{\myat}{{\fontfamily{pag}\selectfont @}}

\noindent E-mail:
\href{paolo.meda@ge.infn.it}{$^a$paolo.meda\myat ge.infn.it},
\href{pinamont@dima.unige.it}{$^b$pinamont\myat dima.unige.it}
\href{simoneroncallo1996@gmail.com}{$^c$simoneroncallo1996\myat gmail.com}, \\
\href{nino.zanghi@ge.infn.it}{$^d$nino.zanghi\myat ge.infn.it},

\normalsize

\par

\rm\normalsize

\par 
\bigskip

\rm\normalsize 

\bigskip

\noindent 
\small 
{\bf Abstract.}
We study the evaporation of a four-dimensional spherically symmetric black hole formed  in a gravitational collapse. We analyze the back-reaction of a massless quantum scalar field conformally coupled to the scalar curvature by means of the semiclassical Einstein equations. We show that the evaporation is linked to an ingoing negative energy flux at the dynamical horizon and that this flux is induced by the quantum matter trace anomaly outside the black hole horizon whenever a suitable averaged energy condition is satisfied. For illustrative purposes, we evaluate the negative ingoing flux and the corresponding rate of evaporation in the case of a null radiating star described by the Vaidya spacetime.

\bigskip

\section{Introduction}
\label{sec:intro}
In semiclassical approximation of quantum gravity, matter is described by quantum fields propagating on a curved classical background in such a way that, given a state  $\omega$ of  the quantum field, the back-reaction on the classical background is determined by the semiclassical Einstein equations
\begin{equation}
\label{eq:SCE}
	G_{\mu\nu} = 8\pi \expvalom{T_{\mu\nu}}
\end{equation}
in units convention $G = c = \hbar = 1$. Here $G_{\mu\nu}$ is the usual Einstein tensor and $\expvalom{T_{\mu\nu}}$ is the mean expectation value of the quantum stress-energy tensor in the state $\omega$. The classical background is assumed to be a four-dimensional globally hyperbolic spacetime $(\mathcal{M},g)$, with $\mathcal{M}$ a smooth manifold and $g$ a Lorentzian metric with signature $(-,+,+,+)$. We consider self-consistent solutions in semiclassical gravity. A self-consistent solution is a pair composed by a spacetime metric and a quantum state satisfying \eq \eqref{eq:SCE} at all orders in $\hbar$. Given any such solution, the quantum stress-energy evaluated in $\omega$ gives the Einstein tensor of the spacetime $(\mathcal{M},g)$. Conversely, the Einstein tensor constructed from the metric $g$ yields the expectation value of the quantum stress-energy tensor of the quantum field consistent with $\omega$. It is often argued that solutions of \eq \eqref{eq:SCE} furnish approximations to models of a more fundamental theory of quantum gravity valid in the regime where $R_{\mu\nu\rho\sigma}R^{\mu\nu\rho\sigma} \ll m_P^4$, with $m_P$ denoting the Planck mass, and when the fluctuations of the stress-energy tensor are small \cite{Kuo:1993if}.

In the seminal works \cite{Hawking1974black,Hawking1975}, Hawking showed that a Schwarzschild black hole emits a radiation which can be detected as a flux of particles at large distances from the black hole. This is the well known Hawking radiation, originally obtained assuming no back-reaction and with the quantum matter being in a vacuum state in the asymptotic past. Its power spectrum is thermal with temperature $T_{H}$, the so-called Hawking temperature, given in terms of the black hole mass $M$ by
\begin{equation}
\label{temp}
T_{H} = \frac{1}{8\pi k_B M}
\end{equation}
(in our units convention), with  $k_B$ being the Boltzmann constant. Under these assumptions, the power of the radiation at future infinity is regarded as describing the evaporation of the black hole. In the adiabatic approximation, the value of the total power emitted is equated to the rate of the loss of mass of the black hole, which turns out to be constant and proportional to $M^{-2}$ (see also \cite{Page:1976df}). 

{This standard derivation is defective in various respects. On one hand, the adiabatic approximation is not able to predict the precise form of the horizon of a black hole with non constant mass.  Moreover, it misses the dynamical nature of the process of evaporation. It is important to stress that these defects might be cured, at least in principle, within the framework  of the semiclassical approximation. Indeed, according to \eq \eqref{eq:SCE}, the rate of mass loss  turns out to be proportional to a flux of negative energy across the black hole horizon. Though in the absence of the static symmetry the power radiated at infinity is not directly linked to that negative flux of energy at the horizon (and its precise value  depends sensibly  on the spacetime details), this radiation (if any) can be detected at future infinity. On the contrary, we shall see that the flux of negative energy at the horizon is constrained and forced to be present in the model by the quantum nature of matter in the causal past of the horizon. This is the simple idea we shall exploit to show evaporation.
	
Preliminarily, we observe that concepts like event horizons and its adiabatic changes  are based on global properties of the spacetime. So, in the case of dynamical backgrounds, they need to be replaced by apparent horizons and their evolutions. Indeed, some attempts to study the back-reaction on four-dimensional spherically symmetric black holes have been made, e.g., assuming the geometrical optics approximation, which simplifies the {analysis} to a two-dimensional problem \cite{balbinot:1984hawking,balbinot:1989backreaction}.
	
In our work we rely on the semiclassical Einstein equations \eqref{eq:SCE}, without making further approximations. Without any appeal to global properties or the explicit form of the stress-energy tensor on the horizon, we shall perform a  local analysis of the apparent horizons and show that, in the case of spherically symmetric spacetimes, their dynamics can be constrained by the matter content outside and in the causal past of the black hole. In particular, by considering the stress-energy tensor of a massless, conformally coupled scalar field, and its conservation laws, we shall prove that the evaporation of the black hole is induced by the form of the quantum trace anomaly outside the horizon.
	
This follows from a very natural condition of the quantum state in the causal past and a mild assumption on the energy outside the horizon, which holds in the case of classical matter. The only quantum property of matter which is used in the argument is the form of the trace anomaly, without which no evaporation can occur. We shall now provide some background and a sketch of our analysis.
	
Firstly, we recall that the process of evaporation can be semiclassically explained by the presence of an ingoing flux of negative energy on the horizon, violating the classical null energy condition $T(k,k) \geq 0$ for all null vectors $k$ on the horizon. This flux is responsible of the shrink of the dynamical area of the horizon, which can in general increase or decrease respectively during a formation or evaporation process \cite{Ashtekar:2002ag,Ashtekar:2003hk,Ashtekar:2004cn}. Indeed, such a violation on the horizon is not surprising, since it has been proven that a very large class of pointwise energy conditions are not valid when quantum fields are involved, even in flat spacetime (the Casimir effect, for instance)}.\footnote{See, e.g., \cite{Ford:1996er,Flanagan:1997gn} and the references therein. For instance, it is known that the emission of Hawking radiation and thus the black hole evaporation can be ascribed to the presence of an anomalous trace  in the stress-energy tensor of a quantum matter field near the horizon \cite{Christensen:1977jc,davies1976energy,fulling1976radiation,Bardeen:1981black, Bardeen:2014uaa} (see also \cite{Casadio:2004th} for a discussion about this topic in relation with the AdS-CFT correspondence).}

{Secondly, it is well known that in the framework of quantum fields on curved spacetimes, an anomalous term is present in the expectation value of the trace of the quantum stress-energy tensor in any physically reasonable state. Such an anomalous trace is a local contribution which depends only on the geometry and  the linear equation of motion of the matter field; it arises in any covariant regularization procedure which gives origin to a  covariantly conserved  stress-energy tensor $\wick{T_{\mu\nu}}$, while breaking the classical conformal invariance of $T_{\mu\nu}$ \cite{Wald1977back,Wald:1978pj,birrell1984quantum,Hollands:2004yh,moretti2003comments}. Moreover, although the expectation value of the quantum stress-energy tensor $\expvalom{T_{\mu\nu}}$ is not explicitly available for dynamical black holes (due to the absence of precise control of sufficiently regular states in that context), its trace-anomaly can be evaluated explicitly and independently of the state also for dynamical backgrounds.}\footnote{For some references about the computation of the stress-energy tensor in Schwarzschild, see \cite{Candelas:1980zt,Howard:1984vsc,Anderson:1995sts,Anderson:2020gsq}. For details about the definition of quadratic observables like $T_{\mu\nu}$ and $\phi^2$ as normal ordered fields, see \cite{Wald:1995yp,Hollands:2001nf,Hollands:2001fb,Brunetti:1996mc}. Finally, see \cite{hollands2015quantum} for a general discussion about quantum field theory on curved spacetimes and its applications like Hawking radiation.} 

{Thirdly, black hole evaporation follows from the trace anomaly whenever the stress-energy tensor satisfies a suitable averaged energy condition outside the horizon (see also \cite{Emelyanov:2019stco}). In this extra condition (see below) the pointwise expectation value of a particular component of the stress-energy tensor, is smeared with a suitable strictly positive smooth function supported outside the black hole horizon. Furthermore, this condition needs to hold only for the particular state $\omega$ used in \eq \eqref{eq:SCE} even if it is similar in spirit to other quantum averaged energy conditions, which hold in any quantum state and that have been established in many contexts} \footnote{See, e.g., \cite{Wald:1991xn,Ford:1994bj,Ford:1995gb,Fewster:2002ne,Fewster:2002dp,Brown:2018hym,Fewster:2019bjg,Schlemmer:2008dk,Freivogel:2020hiz} (see also \cite{Kontou:2020bta} for a general review about classical and quantum energy inequalities and further references about this topic).} Finally, we observe that this energy condition can be linked to the geometry described by the semiclassical metric forming with the state $\omega$ a solution of the semiclassical Einstein equations \eqref{eq:SCE}. If we also assume that quantum corrections are negligible outside and in the past of the horizon, we get that the energy condition is then satisfied in known models of gravitational collapse like the Oppenheimer-Snyder and Lema\^itre-Tolman-Bondi models \cite{Griffiths:2009dfa} and it is also compatible with the collapsing matter described by a classical scalar field in the works of Christodoulou \cite{Christodoulou1986:global,Christodoulou1986:self,Christodoulou:1991yfa}.

\medskip

The paper is organized as follows. In \hyperref[sec:spherical]{Section 2} we recollect some geometric aspects of spherically symmetric spacetimes and apparent horizons and we recall Hayward's thermodynamic interpretation of black hole dynamics. In \hyperref[sec:variation-mass]{Section 3} we describe the semiclassical process of evaporation due to the negative ingoing flux on the horizon and provide an equation for the variation of the mass. In \hyperref[sec:trace-anomaly]{Section 4} we show that the quantum trace anomaly of a free massless conformally coupled scalar field drives the evaporation assuming a certain averaged quantum energy inequality, which is also satisfied by the background geometry in most realistic classical models of collapse. As an example, we compute the rate of evaporation in the Vaidya spacetime {and, as a byproduct, we find that the Schwarzschild spacetime cannot be in equilibrium with the back-reaction of any quantum matter field}. \hyperref[sec:concl]{Section 5} contains the conclusions and some possible future developments. The technical details to obtain the equation for the variation of the mass and the proof of Theorem \ref{theo:evap} are collected in the \hyperref[app]{Appendix}.

\section{Spherically symmetric black holes}
\label{sec:spherical}

A spherically symmetric spacetime $(\cM,g)$ is represented by the manifold $\cM = \Gamma \times \mathds{S}^2$, where $\mathds{S}^2$ is the two-dimensional sphere of unital radius and $\Gamma$ is a two-dimensional space normal to $\mathds{S}^2$, and by the metric 
\[
	\d s^2 = g_{\mu\nu} \d x^\mu \d x^\nu = \gamma_{ij} \d x^i \d x^j + r^2\d \theta^2 + r^2\sin^2\theta \d \varphi^2 \ , \qquad \mu,\nu = 0,\dots,3, \quad i,j = 1,2,
\]
where $r^2 \in \cC^{\infty}(\Gamma)$ measures the curvature of each sphere. The two-dimensional spacetime $(\Gamma,\gamma)$ corresponds to the quotient of $\cM$ with respect to the $SO(3)$ group centered at the origin $r=0$. The invariant $\nabla_\mu r \nabla^\mu r$ defines the Misner-Sharp energy
\begin{equation}
\label{eq:Hawking-mass}
	m \doteq \frac{r}{2} \at 1 - \nabla_\mu r \nabla^\mu r \ct ,
\end{equation}
which describes the energy enclosed inside the sphere of radius $r$ (it is a special case of the Hawking mass for the class of spherically symmetric spacetimes) \cite{Misner:1964je, Hawking:1968qt}.

In order to describe an evaporating dynamical black hole, the two-dimensional normal line element $\d \gamma^2$ is often represented in the Bardeen-Vaidya metric \cite{Bardeen:1981black}
\begin{equation}
\label{eq:metric_eddington}
	\d \gamma^2 = -\e^{2\Psi(v,r)} C(v,r) \d v^2 + 2\e^{\Psi(v,r)} \d v \d r, 
\end{equation}
where $v$ is the advanced time and
\begin{equation}
\label{eq:C-mass}
	C(v,r) \doteq 1-\frac{2 m(v,r)}{r}.
\end{equation}
In this parametrization, which corresponds to the advanced Eddington-Finkelstein coordinates in the vacuum case, a radial curve at constant $v$ describes an ingoing null geodesic.

Moreover, any two-dimensional metric is locally conformally flat, so we can choose to parametrize $\d \gamma^2$ in terms of double-null coordinates $(V,U)$ 
\begin{equation}
\label{eq:metric-symmetric}
	\d \gamma^2 = - 2A(V,U) \d V \d U
\end{equation}
with respect to the null normal directions $\pa_V$ and $\pa_U$. The orientation of the spacetime can be also chosen in such a way that $A(V,U) > 0$  and at spatial infinity $\pa_V r >0$, $\pa_U r < 0$. The metric is invariant under any re-parametrization $U \mapsto \tilde{U}(U)$ and $V \mapsto \tilde{V}(V)$, then we can represent the future-directed null normal vector fields as 
\begin{equation}
\label{eq:ell-pm}
	\ell_+ \doteq A^{-1} \pa_V, \qquad \ell_- \doteq \pa_U,
\end{equation} 
which, respectively, describe the outgoing and the ingoing light rays across the spheres that foliates $\cM$. In this parametrization, the vector field $\ell_+$ fulfils the geodesic equation whereas $\ell_-$ is an auxiliary vector. The normalization of $\ell_\pm$ is such that $g_{\mu\nu}\ell_+^\mu \ell_-^\nu = -1$. The local change of coordinates which relates the metrics \eqref{eq:metric_eddington} and \eqref{eq:metric-symmetric} is given by
\begin{subnumcases}{}
	2A \d U = \e^{2\Psi(v,r)}C(v,r) \d v   - 2\e^{\Psi(v,r)} \d r, \label{eq:change_coordinates_a} \\
	\d V = \d v. \label{eq:change_coordinates_b}
\end{subnumcases}
Following the conventions given by Hayward \cite{Hayward:1993mw,Hayward:1994bu,Hayward:1997jp,Hayward:2000ca}, each sphere that foliates $\cM$ is defined to be untrapped, marginal or trapped depending on whether the dual vector $\nabla^\mu r$ is spacelike, lightlike or timelike, respectively. If $\nabla^\mu r$ is future/past-directed, then the sphere is future/past trapped: the past case is related to white holes, whereas the future one to black holes, where both outgoing and ingoing light rays are trapped into the surfaces. An hypersurface foliated by marginal spheres is called a trapping horizon and a trapping horizon is outer, degenerate or inner when $\nabla^2 r > 0$, $\nabla^2 r = 0$ or $\nabla^2 r < 0$, respectively. In the $(V,U)$ foliation one considers the expansion parameters of the congruences of outgoing/ingoing radial null geodesics
\begin{equation}
\label{eq:expansions}
	\theta_+ \doteq \frac{2}{A r} \pa_V r, \qquad \theta_- \doteq \frac{2}{Ar} \pa_U r.
\end{equation}
Then, a trapping surface is defined as a compact spatial two-surface with $\theta_+\theta_- \geq 0$; it is future/past when $\theta_\pm > 0$ or $\theta_\pm < 0$, respectively, and marginal when $\theta_+= 0$. A trapping horizon is defined as an hypersurface foliated by marginal surfaces; moreover, it is future if $\theta_- <0$ or past if $\theta_- > 0$, outer if $\pa_U \theta_+ < 0$ or inner if $\pa_U \theta_+ >0$. Note that both the expansions are smooth functions outside $r=0$ since $\mathcal{M}$ is smooth. In the framework of black hole physics, an \textit{apparent horizon} is defined to be a future outer trapping horizon satisfying
\begin{equation}
\label{eq:FOTH}
	\theta_+ \eqH 0, \qquad \theta_- \leqH 0, \qquad \pa_U \theta_+ \leqH 0,
\end{equation}
where the subscript $\cH$ labels the evaluation on the apparent horizon. The first two conditions capture the fact that no outgoing ray can escape from $\cH$, different from the ingoing ones which are converging therein; the third condition means that the area of the outgoing congruence is increasing just outside $\cH$ and it is decreasing just inside $\cH$. On spherically symmetric spacetimes, the apparent horizon is the three-dimensional hypersurface
\begin{equation}
\label{eq:apparent-horizon}
	\cH = \left\{(p_\gamma,\Omega) \in \cM: \quad r - 2m = 0 \right\},
\end{equation}
which is also a dynamical horizon according to the definition given in \cite{Ashtekar:2002ag, Ashtekar:2003hk, Ashtekar:2004cn}. Hence, the mass of the black hole $M$ is defined as the Misner-Sharp energy evaluated on the horizon $r_\cH \doteq 2m$
\begin{equation}
\label{eq:M}
	M \doteq m(r_\cH) = \frac{r_\cH}{2}.
\end{equation}
In coordinates $(v,r)$ the apparent horizon is described by the line $r_\cH(v)$ defined by $C(v,r) =0$. Furthermore, the mass $M(v) \doteq m\at v,r_\cH(v)\ct$ is fully determined by the rate of evaporation
\begin{equation}
\label{eq:rate}
	\dot{M}(v) \doteq \pa_v M(v).
\end{equation}

On spherically symmetric spacetimes a preferred notion of ``time'' exists due to the definition of the \textit{Kodama vector} \cite{Kodama:1980ko,abreu2010kodama}
\begin{equation}
\label{eq:kodama}
	K \doteq g^{-1}(\ast dr) = A^{-1} \at \pa_V r \pa_U - \pa_U r \pa_V \ct, 
\end{equation} 
where $\ast$ is the Hodge operator in the space normal to the spheres. {This vector is proportional to the timelike Killing vector $\pa_t$ on static spherically symmetric spacetimes. The Kodama vector is divergenceless even if it is not a Killing field, $\nabla_\mu K^\mu = 0$, and furthermore the Kodama flux ${T^\mu}_\nu K^\nu$ defines a covariantly conserved current for any stress-energy tensor ${T^\mu}_\nu$, namely}
\begin{equation}
\label{eq:kodama-conservation}
	\nabla_\mu ({T^\mu}_\nu K^\nu) =0.
\end{equation}
The Kodama vector is timelike on untrapped spheres, i.e., in the region outside the horizon, and becomes lightlike on a marginal sphere, and eventually it is spacelike on trapped surfaces, i.e., in the interior of the black hole. From the definition of $K$, one obtains also that
\begin{equation}
	\mathscr{L}_K K_\mu = K^\nu \at \nabla_\nu K_\mu - \nabla_\mu K_\nu \ct \eqH \pm \kappa K_\mu,
\end{equation}
where $\mathscr{L}_K$ denotes the Lie derivative along $K$ and
\begin{equation}
\label{eq:kodama-box}
	\kappa \doteq \frac{1}{2} \gamma^{ij}\nabla_i\nabla_j r.
\end{equation}
Hence, \eq \eqref{eq:kodama-box} corresponds to the definition of the surface gravity for a dynamical black hole and it reduces to the standard one in the case of a Killing vector field (for a discussion about the different definitions of $\kappa$ in literature, see \cite{Vanzo:2011wq}). Thus, $\kappa$ represents the gravitational acceleration detected along the black hole horizon. From the definition \eqref{eq:kodama-box}, a trapping horizon is outer, degenerate or inner when $\kappa$ is positive, null or negative, respectively; in particular $\kappa > 0$ along an apparent horizon like $\cH$.

A thermodynamic interpretation of the evolution of the mass along apparent horizons has been given by Hayward \cite{Hayward:1997jp}, who proved the first law of black hole (thermo)dynamics for spherically symmetric black holes,
\begin{equation}
\label{eq:first-law}
	\pa_k  m = \frac{\kappa}{8\pi} \pa_k \cA + w \pa_k \cV.
\end{equation}
Here, $\pa_k f= k \cdot \nabla f$ denotes the derivative along any vector field $k$ tangent to the horizon, $m$ is the mass \eqref{eq:Hawking-mass}, $\kappa$ is the surface gravity \eqref{eq:kodama-box}, $\cA = 4\pi r^2$ is the area, $\cV = \frac{4}{3} \pi r^3$ the volume and $w \doteq -\frac{1}{2} \gamma^{ij}T_{ij} = A^{-1} T_{UV}$ is the work density done by the matter field on the horizon. Furthermore, assuming the Einstein equation $w = (8\pi)^{-1}{R_\theta}^\theta$,
\begin{equation}
\label{eq:kappa-w}
	 \kappa = \frac{m}{r^2} - 4\pi r w.
\end{equation}
As in the case of static black holes, the thermodynamic interpretation of \eq \eqref{eq:first-law} can be made precise only if the surface gravity is proportional to an actual temperature. This is really the case of a static black hole, where
$\kappa/2\pi$, in natural units, equals the Hawking temperature of Hawking radiation observed at future infinity \cite{Hawking1975,fredenhagen1990derivation,KAY199149,Wald:1999vt}. In the dynamical case, it is possible to show that the very same temperature can be seen in the tunnelling probability of matter across dynamical horizons. A derivation of this fact involving the WKB approximation for the one particle excitations can be found in \cite{DiCriscienzo:2007pcr, Hayward:2008jq}, based on the ideas presented in \cite{Parikh:1999mf}. Moreover, another derivation focusing on the properties of states for quantum fields near apparent horizons is presented in \cite{Kurpicz:2021kgf}. The latter observation enforces the statement that $\kappa$ defined in \eq \eqref{eq:kodama-box} must be a positive quantity at least near the apparent horizon.

\section{{Variation of the mass, energy fluxes, and their constraints from the causal past}}
\label{sec:variation-mass}

{Contrary to the case of a static null event horizon, an apparent horizon like the one in \eq \eqref{eq:apparent-horizon} can evolve as a dynamical trapping hypersurface in a process of black hole formation or evaporation, under the influence of the matter. Thus, one can also infer the dynamics of the mass of the black hole \eqref{eq:M}, which can respectively increase or decrease according to the evolution of the function $r_\mathcal{H}$. In this paper, the dynamical evolution of the apparent horizon and of the black hole mass is analyzed from a local point of view, assuming the (semiclassical) Einstein equations as the only dynamical equation governing the interplay between matter and geometry. We shall not refer to any asymptotic effect at large distances from or in the future of the black hole, because such a global approach would require the knowledge of the entire history of the spacetime. 	

Let us assume that the matter content is fully described by a generic stress-energy tensor $T_{\mu\nu}$}. According to the definition of trapping horizon, the local dynamics of $\cH$ can be related to the evolution of the expansion parameter $\theta_+$ given in \eq \eqref{eq:expansions} along an outgoing null geodesic. Denoting $\d / \d  \lambda \doteq \ell_+^\mu \nabla_\mu$ the directional derivative along $\ell_+$, the {equation $G_{VV} = 8\pi T_{VV}$} reads
\begin{equation}
\label{eq:Raychaudhuri}
	\frac{\d \theta_+}{\d \lambda} = -\frac{1}{2} \theta_+^2 - 8\pi {{T_{\mu\nu}}} \ell_+^\mu \ell_+^\nu,
\end{equation}
which represents the Raychaudhuri equation for the null affine-parametrized outgoing geodesics congruence \cite{Wald:1984rg} (it is sometimes referred also as the Landau-Raychaudhuri equation). {If the stress-energy tensor is associated to classical matter, the null energy condition $T_{\mu\nu} k^\mu k^\nu \geq 0$ holds for any null vector $k^\mu$ and hence $\d \theta_+ / \d \lambda \leq 0$}. Thus, assuming the initial condition $\theta_+(V_0) = 2/(Ar) > 0$ at the beginning of the collapse, there must exist a region where $\theta_+ = 0$ for $V > V_0$, namely that a trapped surface has formed during the gravitational collapse. On the other hand, when evaluated on the apparent horizon, {where $\theta_+ = 0$,} \eq \eqref{eq:Raychaudhuri} reduces to
\begin{equation}
\label{eq:Raychaudhuri-TVV}
	\frac{\d \theta_+}{\d \lambda} \eqH - 8\pi \frac{{T_{VV}(r_\cH)}}{A^2},
\end{equation}
{namely the evolution of the apparent horizon is directly related to the ingoing energy flux $T_{VV}$ evaluated on the horizon. If such component has a quantum nature, then it can violates the classical null energy condition, and hence {if $T_{VV}$ is negative on the horizon, it happens that} $\d \theta_+ / \d \lambda > 0$}. Therefore, in this case, the trapped surface formed during the collapse tends to disappear, namely it evaporates. This process of evaporation of the horizon makes manifest as a loss of the black hole mass given in \eq \eqref{eq:M}. Given a one-dimensional portion of horizon $\delta \cH \subseteq \pi_\gamma (\cH)$, where $\pi_\gamma: \Gamma \times \mathds{S}^2 \to \Gamma$ denotes the natural projection on the first pair of coordinates, let us define the variation of mass of the black hole on $\delta \cH$  
\begin{equation}
\label{eq:dM}
	\Delta M \doteq \int_{\delta \cH} \d m.
\end{equation}
In coordinates $(v,r)$, $\delta \cH$ is the line enclosed between two arbitrary points $(v_P,r_\cH(v_P))$ and $(v_Q,r_\cH(v_Q))$ in the $(v,r)$ plane. On $\cH$, the relation $\d r = (1-2\pa_r m)^{-1} 2\pa_v m \d v$ holds and 
\eq \eqref{eq:Raychaudhuri-TVV} becomes a dynamical law for the rate $\dot{M}$ defined in \eq \eqref{eq:rate}. Actually, after rescaling $v$ so that $\Psi(v,r_\cH) = 0$, both \eq \eqref{eq:Raychaudhuri-TVV} and {$G_{vv} r^2 \eqH 8\pi T_{vv} r^2$} read
\begin{equation}
\label{eq:TVV-mass}
	\dot{M}(v) = \cA_\cH (v) {T_{vv}}(v,r_\cH(v)),
\end{equation}
where $\cA_\cH \doteq 4\pi r^2_\cH$ denotes the area of the horizon. Thus, taking into account the surface gravity \eqref{eq:kodama-box} on $\cH$,
\begin{equation}
\label{eq:dM-Tvv}
	\Delta M = \int_{v_P}^{v_Q} \frac{\dot{M}}{4M \kappa} \d v = 4\pi \int_{v_P}^{v_Q} \frac{M}{\kappa} {T_{vv}}(r_\cH) \d v.
\end{equation}
Since $\kappa \geqH 0$ and ${T_{vv}} \eqH {T_{VV}}$, both the rate of evaporation \eqref{eq:rate} and the variation of the mass \eqref{eq:dM} are negative when ${T_{VV}(r_\cH)} < 0$. Moreover, given the vector $n_{\cH} \doteq g^{\mu\nu}\pa_{\nu} \cH \pa_\mu$ normal to the apparent horizon in the $(v,r)$ plane, then
\begin{equation}
\label{eq:normal-vector}
	g_{\mu\nu} n_{\cH}^\mu n_{\cH}^\nu \eqH -16 M \kappa \dot{M}.
\end{equation}
So, if $\Delta M < 0$ then $\dot{M} < 0$ and $n_{\cH}$ is spacelike, hence $\cH$ is a timelike surface, namely the corresponding black hole is evaporating.

{It is actually difficult to evaluate or to estimate directly the negative energy flux across the horizon.
However, it turns out that some constraint for $T_{VV}(r_\cH)$ and thus for $\Delta M$ can be given in terms of the matter stress-energy tensor evaluated in the causal past and outside the black hole horizon.
This shall be done by applying the divergence theorem (Stokes' theorem) to the currents obtained contracting the stress-energy tensor with suitable vector fields and constructed in such a way that the corresponding flux across $\delta \cH\times \mathds{S}^2$ coincides with $4\pi \Delta M$. In the next section we shall discuss how some components of the stress-energy tensor can be constrained outside the black hole horizon with the trace anomaly by employing this analysis.
More precisely, the following currents can be obtained contracting the stress-energy tensor with the gradient $\nabla r$ and with the Kodama vector \eqref{eq:kodama}}: 
\begin{align}
	J^\mu_r &\doteq {T^{\mu\nu}} \nabla_\nu r = J^\mu_1 + J^\nu_2 = {T^{U\mu}} \pa_U r + {T^{V\mu}} \pa_V r \label{eq:current-r},\\
	J^\mu_K &\doteq {T^{\mu\nu}} K_\nu = J^\mu_1 - J^\nu_2 = {T^{U\mu}} \pa_U r - {T^{V\mu}} \pa_V r\label{eq:current-k}.
\end{align}
Denoting by $\nabla \cdot J \doteq \nabla_\mu J^\mu$ the divergence of the current $J$, from \eq \eqref{eq:kodama-conservation} it follows $\nabla \cdot J_K = 0$, which implies that $\nabla \cdot J_1 = \nabla \cdot J_2$ and $\nabla \cdot J_1 = \frac{1}{2} \nabla \cdot J_r$. {Moreover, on the horizon $J_2 = 0$ because $\pa_V r =  0$, hence the flux across $\delta \cH\times \mathds{S}^2$ of $J_r$ and $J_K$ coincide.
A direct analysis of this flux shows that the flux across $\delta \cH \times \mathds{S}^2$ of both $J_r$ and $J_K$ coincides (up to a factor $4\pi$) with $\Delta M$ given in \eq \eqref{eq:dM}.

The domain over which the divergence theorem is applied to obtain $\Delta M$ is actually spherically symmetric and it has the form $D\times\mathds{S}^2$, where $D$ is a suitable portion of $\Gamma$.} To define $D$ more precisely, consider
\begin{align}
	S_0 &\doteq\{(V,U,\theta,\varphi)\in \mathcal{M} \mid V=V_0, \ U=U_0\} \label{eq:defS0}, \\
	S_1 &\doteq\{(V,U,\theta,\varphi)\in \mathcal{M} \mid V=V_1, \ U=U_0\} \label{eq:defS1}, 
\end{align}
where $U_0$, $V_0$ and $V_1>V_0$ are chosen such that both $S_0$ and $S_1$ lie outside the apparent horizon $\cH$ and in such a way that $(V_1,U_2) \times \mathds{S}^2$ is contained on $\mathcal{H}$ for some $U_2$. Consider now $\delta \cH \times \mathds{S}^2$ the portion of $\cH$ which intersects $J^+(S_0) \cap (\mathcal{M}\setminus I(S_1))$, and denote by $P_\cH = (V_1,U_2)$ and $Q_\cH = (V_0,U_1)$ the extreme points of $\delta \cH$ in the $(V,U)$ plane. The domain $D\times\mathds{S}^2$ is then obtained by considering the portion of $J^+(S_0)  \cap (\mathcal{M}\setminus I(S_1))$ which lies outside the apparent horizon $\cH$. If $\cH$ is spacelike or null,
\begin{equation}
\label{eq:D-space}
	D\times \mathds{S}^2 \doteq   J^+(S_0)  \cap J^-(\delta \cH \times \mathds{S}^2),
\end{equation}
while, if $\cH$ is timelike, 
\begin{equation}
\label{eq:D-time}
	D\times \mathds{S}^2 \doteq   J^+(S_0)  \cap O, 
\end{equation} 
where $O$ is the portion of $J^-(\delta \cH \times \mathds{S}^2)$ which lies outside the horizon. With these definitions, $\rho_0, \delta_0, \gamma \in \partial D$ denote the one-dimensional curves in the $(V,U)$ plane between $(V_0,U_1)$ and $(V_0,U_0)$, $(V_0,U_0)$ and $(V_1,U_0)$, $(V_1,U_0)$ and $(V_1,U_2)$, respectively. See Figure \ref{fig:D} for a representation of $D$.
\captionsetup[figure]{labelfont={bf},labelformat={default},name={Figure}}
\begin{figure}[ht] \centering{\includegraphics[scale=0.85]{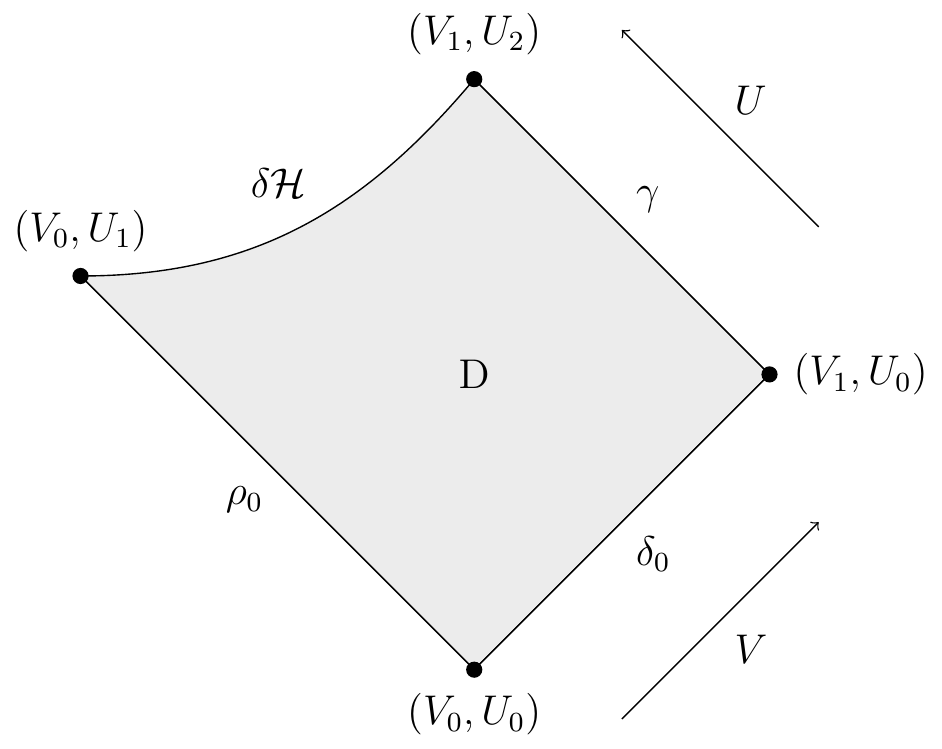}}
	\caption{\footnotesize Picture of the domain of integration $D$ given a spacelike portion $\delta \cH$ of the apparent horizon. In this case, $D\times \mathds{S}^2 = J^{+}((V_0,U_0) \times \mathds{S}^2) \cap J^{-}(\delta \cH\times \mathds{S}^2)$, with $U_2 < U_1$. The initial data are posed on the curves $\rho_0$ and $\delta_0$ of the boundary $\partial D$.}
	\label{fig:D}
\end{figure}

{With this definition of $D$ and with $S_0$ given by \eq \eqref{eq:defS0}, let us consider a stress-energy tensor $T_{\mu\nu}$ which respects the spherical symmetry and which satisfies the following initial conditions at the boundary of $J^+(S_0)$:
\begin{equation}
\label{eq:initial-condition}
	T_{\mu\nu}(p) =0, \qquad p\in \partial J^+(S_0).
\end{equation}
With the choice of $S_0 \in \mathscr{I}^-(\Gamma)$, such an initial condition states that there is no influence of the matter in the past infinity. Then, applying the divergence theorem (Stokes' theorem) to the current $J_1$ on the domain $D\times \mathds{S}^2$, following the derivation given in Appendix \ref{app:proof-DeltaM}, one 
obtains that
\begin{equation}
\label{eq:dM-TUV-J}
	{\Delta M} = -\at \cS + \cW \ct,
\end{equation}
where
\begin{align}
	\cS &\doteq 2\pi \int_{D} \nabla \cdot J_r \d \cV_D \label{eq:flux-S}, \\
	\cW &\doteq 4\pi \int_{\gamma} \frac{T_{UV}r^2}{A} (-\pa_U r) \d U \label{eq:flux-W}.
\end{align}
Here $\d \cV_D \doteq A r^2 \d V \wedge \d U$ denotes the volume form  on $D$ in the $(V,U)$ plane. In particular, according to \cref{eq:dM-TUV-J,eq:flux-S,eq:flux-W}, $\cS$ is the matter source inside the domain $D$, whereas $\cW$ is related to the component of $T_{\mu\nu}$ that is associated to the work $w$ done by the matter when evaluated on $\cH$, in view of Hayward's first law \eqref{eq:first-law}. In the next section we shall see how the quantum trace anomaly shown by the stress-energy tensor outside the black hole horizon forces $\mathcal{S}$ to be positive. We remark in passim that applying the divergence theorem to the current $J_K$ and comparing the result with \eq \eqref{eq:dM-TUV-J} yield $\cR = \cS$, where
\begin{equation}
\label{eq:flux-R}
	\cR \doteq 4\pi \int_{\gamma} \frac{T_{UU}r^2}{A} \pa_V r \d U
\end{equation}
is related to the outgoing energy flux $T_{UU}$ across $\gamma$. Hence, in the case of stress-energy tensors satisfying $\cW = 0$ it follows that $\Delta M = - \cR$. Finally, as before, the constraints on $\mathcal{S}$, that shall be imposed by the trace anomaly in the next section, force also the outgoing energy flux across $\gamma$ to be positive.}

\section{{Evaporation induced by the quantum trace anomaly}}
\label{sec:trace-anomaly}

In the previous section, \eq \eqref{eq:dM-TUV-J} showed that the ingoing energy flux $T_{VV}$ on the horizon together with $\Delta M$ are constrained by the matter content outside and in the causal past of the horizon, encoded in the source $\cS$ and in the flux $\cW$ given in \cref{eq:flux-S,eq:flux-W} on the domain $D \times \mathds{S}^2$. Now, we shall see that a negative ingoing flux on the horizon, and thus the evaporation, can be obtained considering the non-vanishing trace anomaly of a quantum stress-energy tensor $\expvalom{T_{\mu\nu}}$. Actually, this anomalous trace forces $\cS$ to be positive, and thus $\Delta M$ to be negative according to \eq \eqref{eq:dM-TUV-J}, provided that an auxiliary averaged energy condition is also assumed to control $\cW$.

For a free massless conformally coupled scalar field $\phi$, the non-vanishing trace of the quantum stress-energy tensor $\expvalom{T_{\mu\nu}}$ does not depend on the choice of the quantum state $\omega$, but it is fixed by the geometry of the spacetime to be equal to the quantum trace anomaly. In four dimensions, it reads  
\begin{equation}
\label{eq:T-trace}
	\expvalom{{T_{\rho}}^{\rho}} = \lambda \at {C_{\alpha\beta}}^{\gamma\delta}{C_{\gamma\delta}}^{\alpha\beta} + {R_\mu}^\nu {R_\nu}^\mu - \frac{1}{3}R^2 \ct,
\end{equation}
where
\begin{equation}
\label{eq:lambda}
	\lambda \doteq \frac{1}{720 (4\pi^2)},
\end{equation}
${C_{\alpha\beta}}^{\gamma\delta}$ is the Weyl tensor, ${R_\mu}^\nu$ the Ricci tensor and $R$ the Ricci scalar. The $\square R$ appearing e.g. in \cite{Wald:1978pj} has been cancelled from eq. \eqref{eq:T-trace} by carefully choosing the renormalization freedoms inside the definition of $\expvalom{{T_{\rho}}^{\rho}}$. As discussed in \cite{Hollands:2001nf,Hollands:2004yh}, there is always enough freedom to remove higher order derivatives at the level of the trace. Actually, higher-order derivatives always appear in any renormalized quantum stress-energy tensor (see, e.g., the references given in the \hyperref[sec:intro]{Section 1} about the computation of the stress-energy tensor in the Schwarzschild spacetime). These higher order derivatives could be eliminated at least at the level of the dynamical equations by imposing the consistency of the semiclassical theory as an expansion in $\hbar$ (see, e.g., \cite{Simon:1991hd}). As in this paper we work only with the explicit form of the trace, we do not need to follow that way of reasoning and we simply use the renormalization freedom to remove higher order derivatives from the trace. 
	
A quantum averaged weak energy condition is usually a non-local constraint of the form
\begin{equation}
\label{eq:energy-condition}
	\lim_{\lambda \to \infty} \inf \int f(t/\lambda)^2 \expval{\wick{T_{\mu\nu}}}_\omega k^\mu  k^\nu (\gamma(t)) \d t \geq 0
\end{equation}
for any quantum state $\omega$ where $\expvalom{T_{\mu\nu}}$ can be evaluated. Here, $k^\mu$ is the tangent vector to the affine-parametrized timelike/null geodesic $\gamma(t)$ and $f$ is a real-valued smooth function having compact support on the domain of $\gamma$. It has been shown that this class of averaged energy conditions are valid both in flat and globally hyperbolic spacetimes for some values of the coupling parameter $\xi$, including the conformally coupled case (see \cite{Kontou:2020bta} and references therein). Moreover, conditions like \eq \eqref{eq:energy-condition} can also hold outside the limit and employing a specific sampling function for a restricted class of vacuum-like reference states. The sampling functions are often positive and smooth everywhere in their domain and also decay sufficiently fast at infinity. For instance, some common choices are Gaussian functions or compactly supported test functions having exponential decay in Fourier space, see, e.g., \cite{Fewster:2010mc,Fewster:2020sta,Wu:2021ktn} and references therein. In our case, we shall employ an exponential smooth function $f(V,U)$ constructed from the geometry of the background, whose role is to tame those inside $\expvalom{{T_{\rho}}^{\rho}}$ which do not contribute to a negative variation of the mass $\Delta M$. Thus, the following theorem holds.

\begin{theorem}
	\label{theo:evap}
	Consider a free quantum massless, conformally coupled scalar field $\phi$ propagating on a spherically symmetric dynamical background, whose metric, expressed according to \eq \eqref{eq:metric-symmetric}, solves the semiclassical Einstein equations \eqref{eq:SCE} for a quantum state $\omega$. Suppose that $\omega$ is such that it makes the initial conditions stated in \eq \eqref{eq:initial-condition} valid for the quantum stress-energy tensor $\expvalom{T_{\mu\nu}}$. Let 
	\begin{equation}
	\label{eq:weight-function}
		f(V,U) \doteq f_0(V) \exp(-8\pi \lambda \beta(V,U)),
	\end{equation}		
	where $\beta(V,U)$ is any solution of $\partial_U\beta(V,U)= \frac{2}{r} \partial_V r {R_U}^V$, and $f_0(V) = \exp{- k (V-V_0)}$, for any $V \geq V_0$, is an exponentially decreasing function with a sufficiently large $k>0$. Let ${\Delta M}$ be given by \eq \eqref{eq:dM}. If
	\begin{equation}
	\label{eq:AEC}
		\int_{U_0}^{U_\cH} \frac{\expvalom{T_{UV}} r^2}{A} f(V,U) A \d U \geq 0
	\end{equation}
	in the domain $\text{D} \times \mathds{S}^2$ defined in \cref{eq:D-space,eq:D-time}, for any $V \in [V_0,V_1]$, (i.e., the integral in the left-hand side of the inequality is taken along any ingoing radial null curve connecting the initial point $(V,U_0)$ and $(V,U_\cH) \in \delta \cH$),  then, ${\Delta M} < 0$, namely the evaporation occurs along $\delta \cH$.
\end{theorem}
\begin{proof}
	See the Appendix \ref{app:proof-evap}
\end{proof}

For a qualitative behaviour of $f(V,U)$ in some special cases, see Figure \ref{fig:f}.

\captionsetup[figure]{labelfont={bf},labelformat={default},name={Figure}}
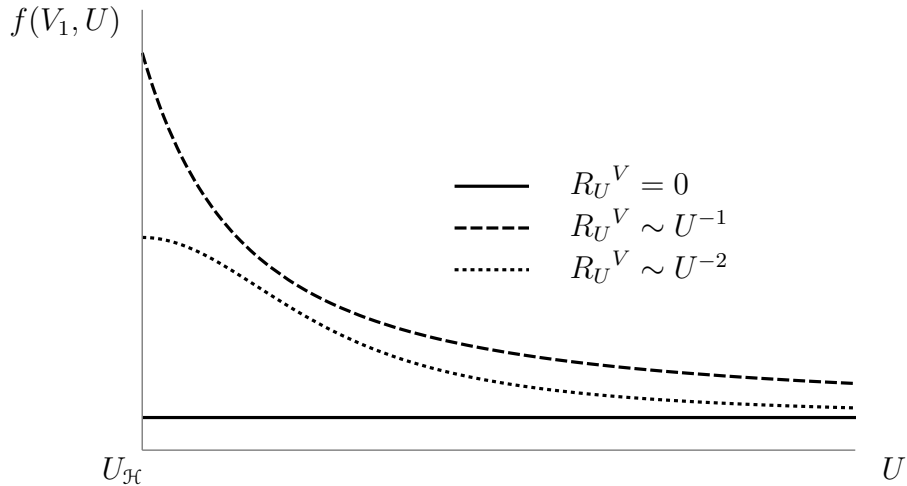
\begin{figure}[ht] 
	
	\begin{picture}(250,160)(0,0)
	
	\put(125,0){\includegraphics[scale=0.75]{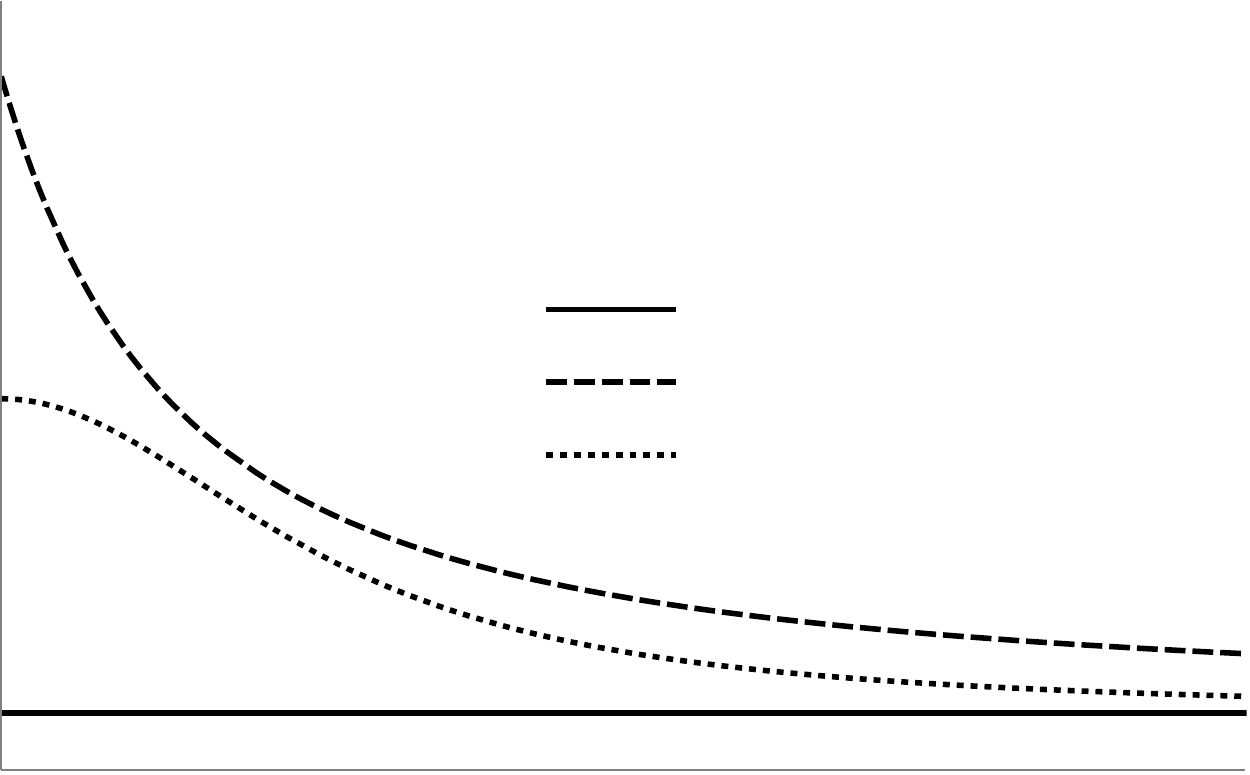}}
	
	\put(75,160){\large $f(V_1,U)$}
	
	\put(405,-10){\large $U$}
	
	\put(110,-10){\large $U_\cH$}
	
	\put(287,98){\large ${R_U}^V = 0$}
	
	\put(287,82){\large ${R_U}^V \sim U^{-1}$}
	
	\put(287,66){\large ${R_U}^V \sim U^{-2}$}
	
	\end{picture}
	
	\bigskip
	
	\caption{\footnotesize Plots of the qualitative behaviour of the smearing function $f(V,U)$ \eqref{eq:weight-function} at fixed $V = V_1$ for special choices of spacetime geometries. In the Schwarzschild and in the Vaidya spacetimes (${R_\theta}^\theta = {R_U}^V = 0$), $f$ is trivially a positive constant. In spacetimes where ${R_\theta}^\theta = 0$ which are asymptotically flat (i.e., $\pa_V r / r$ decays as $U^{-1}$ and ${R_U}^V$ at least as $U^{-1}$ for large $U$), $f$ is bounded and greater than a strictly positive constant, and it approaches that constant for large $U$. In this respect, $f$ is similar to the sampling functions entering usual quantum averaged weak energy conditions (see, e.g., \cite{Fewster:2020sta}).}
	\label{fig:f}
\end{figure}

At this stage, some remarks can be made about the formulation of Theorem \ref{theo:evap}. In this result, the quantum trace anomaly indeed drives the evaporation of the spherical black hole, because in the case of classical free matter field the initial conditions \eqref{eq:initial-condition} would imply that the stress-energy tensor vanishes on $\pa D \times \mathds{S}^2$. Thus, \eq \eqref{eq:AEC} would hold trivially and $\Delta M = 0$, namely $\delta \cH$ would be stable under the influence of the matter field. Moreover, the quantum averaged energy condition stated in the inequality \eqref{eq:AEC} is compatible with the thermodynamic interpretation of $w = A^{-1} {T_{UV}}$ given by Hayward as the work done by the matter on the horizon, which is expected to be positive at least on its average. Also, we believe that this averaged condition can be formulated in more general terms and for a larger class of smooth functions $f(V,U)$, once a sufficiently well-behaved state $\omega$ has been chosen on spherically symmetric spacetimes. For instance, we can expect that such reference state fulfils also the averaged energy condition $\cR > 0$, where $\cR$ was defined in \eq \eqref{eq:flux-R}, with $\pa_V r$ as smearing function. In this case, the quantum outgoing flux $\expvalom{T_{UU}}$ could be interpreted as Hawking radiation emitted from the evaporating $\delta \cH$ and sourced by the trace anomaly inside $\cS$. Unfortunately, the lack of control on the evolution of a quantum state $\omega$ which was a vacuum in the past, namely satisfying \eq \eqref{eq:initial-condition}, prevents us to formulate explicitly a general quantum energy condition compatible with all the previous statements.
	
However, we expect that the condition \eqref{eq:AEC} is fulfilled at least in an approximate way in the causal past. Actually, classical solutions are approximately valid in semiclassical gravity, because quantum corrections are small outside the horizon. Under this approximation, the condition \eqref{eq:AEC} is often satisfied - even pointwise - in the most realistic spherically symmetric models of collapse, where the classical matter sourcing the background fulfils the dominant energy condition $A^{-1} T_{UV} \geq 0$. As examples, we can think at the Lema\^itre-Tolman-Bondi models like the Oppenheimer-Snyder solution, where the collapse is driven by an (in)homogeneous spherical cloud of dust at zero pressure satisfying the weak energy condition, see, e.g., \cite{Griffiths:2009dfa} and references therein. Furthermore, in Christodoulou's work about the collapse in the case of matter described by a classical scalar field \cite{Christodoulou1986:global,Christodoulou1986:self,Christodoulou:1991yfa} the collapsing matter is described by a classical massless scalar field $\phi$ that is invariant under rotations and having stress-energy tensor given by  $T_{\mu\nu} = \pa_\mu\phi\pa_{\nu}\phi -\frac{1}{2}g_{\mu\nu}\pa_\rho \phi \pa^\rho \phi$. Hence, $T_{UV} =0$ and the dominant energy condition holds again.

\medskip

There are two special backgrounds fulfilling $T_{UV} = 0$ from the Einstein equation $G_{UV} = 8\pi T_{UV} = 0$, namely the Schwarzschild and the Vaidya spacetimes. The former describes a static spherically symmetric black hole, while the latter defines the geometry outside a null radiating star, see, e.g., \cite{Griffiths:2009dfa} and references therein. In Bardeen-Vaidya parametrization \eqref{eq:metric_eddington}, the Schwarzschild and the Vaidya metrics are obtained by choosing $m \in \mathds{R}^+$ and $m=m(v)$, respectively, and by fixing $\Psi(v,r) = 0$. In these cases, the semiclassical regime where $R_{\mu\nu\rho\sigma}R^{\mu\nu\rho\sigma} \ll m_P^4$ is always satisfied for $m / m_P \gg (3/4)^{1/4} \simeq 0.93$, which holds for astrophysical masses (for a solar mass, $M_{\astrosun} / m_P \simeq 10^{38}$). In the Vaidya spacetime, the trace anomaly reads
\begin{equation}
\label{eq:trace-Vaydia}
	\expvalom{{T_{\rho}}^{\rho}} = 48 \lambda \frac{M(v)^2}{r^6},
\end{equation}
where $\lambda$ is the coefficient \eqref{eq:lambda}. In this case the rate of evaporation can be directly computed using \eq \eqref{eq:TVV-mass} after evaluating the negative ingoing flux $\expvalom{T_{VV}}$ on $\cH$. To obtain $\expvalom{T_{VV}}$, we employ the conservation equation $\nabla^\mu\expvalom{T_{\mu V}} = 0$, which yields
\begin{equation*}
\label{eq:flux_diff_V}
	-\frac{1}{A r^2} \pa_U \at \expvalom{T_{VV}} r^2 \ct - \frac{1}{r^2} \pa_V \at A^{-1}\expvalom{T_{UV}} r^2 \ct -2 \expvalom{{T_{\theta}}^{\theta}} \frac{\pa_V r}{r} = 0,
\end{equation*}
where 
\begin{equation}
\label{eq:trace-relation}
	\expvalom{{T_{\theta}}^{\theta}} = \frac{\expvalom{T_{UV}}}{A} + \frac{1}{2} \expvalom{{T_\rho}^\rho}.
\end{equation}

It is usually very challenging to evaluate the renormalized quantum stress-energy tensor on a state which is a vacuum state in the past. Furthermore, contrary to its classical counterpart (see, e.g., \cite{davies1976energy} for the two-dimensional case), $\expvalom{T_{UV}}$ is expected not to vanish in a generic quantum state $\omega$, and hence Vaidya spacetime is not expected to be a full solution of the semiclassical equations. Here, we shall assume for simplicity that there exists a quantum state $\omega$ in which $\expvalom{T_{UV}} = 0$ and which makes Vaidya spacetime a semiclassical solution outside the horizon. With this assumption, the ingoing flux fulfils the following differential equation in $(v,r)$ coordinates
\[
	\pa_r (\expvalom{T_{VV}} r^2) = 24 \lambda M(v)^2 \at \frac{1}{r^5} - \frac{2M(v)}{r^6} \ct.
\]
Integrating in $(r_\cH, \infty)$, imposing the initial condition $\expvalom{T_{VV}}r^2 \rightarrow 0$ when $r \rightarrow \infty$, and changing sign, the ingoing flux reads
\begin{equation}
\label{eq:evaporation-Vaidya-TVV}
	\expvalom{T_{VV}} r^2 \eqH -24 \lambda M^2 \int_{r_\cH}^{+\infty} \at \frac{1}{r^5} - \frac{2M(v)}{r^6} \ct \d r = -\frac{3\lambda}{40M(v)^2}.
\end{equation}
Hence, the rate of evaporation obtained from \eq \eqref{eq:TVV-mass} is
\begin{equation}
\label{eq:evaporation-Vaidya-rate}
	\dot{M}(v) = -\frac{3 \pi \lambda}{10M(v)^2}.
\end{equation}
\Eq \eqref{eq:evaporation-Vaidya-rate} is an ordinary differential equation with respect to $v$ and it can be integrated by separation of variables, yielding the evaporation law
\begin{equation}
\label{eq:evaporation-Vaidya-law}
	M^3(v) = M_0^3 - \frac{9 \pi \lambda}{10}(v-v_0),
\end{equation}
where we have defined the total initial mass $M_0 \doteq M(v_0)$ at the initial time $v_0$. Thus, the evaporation process is completed in the time interval $\Delta v \doteq v - v_0 = 10M_0^3/(9\pi \lambda)$. Moreover, according to the first law given in \eqref{eq:first-law} for the Vaidya spacetime, the negative rate \eqref{eq:evaporation-Vaidya-rate} induces also a shrink of the area of the horizon $\cA_\cH$, whose rate of variation $\dot{\cA}_\cH(v) \doteq \pa_v \cA(v,r_\cH(v))$ is governed by $\dot{M}(v) = (\kappa / 8\pi) \dot{\cA}_\cH(v)$. Hence, a negative variation of the Wald-Kodama dynamical entropy $S_\cH \doteq \cA_\cH / 4$ \cite{Ashworth:1998uj,Hayward:1999347} holds, namely
\begin{equation}
\label{entropy-Vaidya}
	\frac{\d S_\cH(v)}{\d v} = - \frac{12 \pi^2 \lambda}{5 M(v)}.
\end{equation}

Therefore, it follows that the Schwarzschild spacetime is not a solution of the semiclassical Einstein equations, because the variation of the mass trivially vanishes in the case of constant mass. This prevents to obtain an equation for the rate like \eq \eqref{eq:evaporation-Vaidya-rate} in the static case. Hence, an eternal black hole cannot be in equilibrium with the back-reaction of any quantum matter field outside the horizon which is in a vacuum state in the asymptotic causal past.

The only quantum property of matter which was used to obtain black hole evaporation is the anomalous contribution to the trace of the quantum matter stress-energy tensor. Hence, an immediate generalization of the foregoing argument may be carried out by extending the analysis to arbitrary massless conformally coupled fields, after modifying the coefficient $\lambda$ inside the trace anomaly given in \eq \eqref{eq:T-trace}. In the general case, the four-dimensional anomalous trace is given by $\expvalom{{T_{\rho}}^{\rho}} = b_F F + b_G G$, where $F = C_{\alpha\beta\gamma\delta}C^{\alpha\beta\gamma\delta}$ is the square of the Weyl tensor, $G = R_{\mu\nu\rho\sigma} R^{\mu\nu\rho\sigma} - 4R_{\mu\nu}R^{\mu\nu} + R^2$ is the Euler density and $b_F$, $b_G$ are coefficients depending on the numbers of particles $n_s$ of spin $s$. For the explicit values of $b_F$ and $b_G$, see \cite{birrell1984quantum}. Arguably, a generalization of the Theorem \ref{theo:evap} can be obtained for arbitrary fields after choosing properly the coefficients inside $\expvalom{{T_{\rho}}^{\rho}}$. Further generalizations of the analysis presented in this paper beyond the spherically symmetric case are harder to obtain.

\section{Conclusion}
\label{sec:concl}
The understanding of the mechanism that leads to the evaporation of a (spherically symmetric) black hole is totally within the scope of semiclassical gravity. It turns out that the negative variation of the black hole mass is due to a negative ingoing flux on the horizon. Such a flux can be obtained by modelling matter outside and in the causal past of the horizon as a conformally coupled quantum scalar field. This model clearly shows that the key of evaporation is the quantum trace anomaly for suitable vacuum-like initial conditions in the past. Of course, to overcome the poor control on the state-dependent contribution to the stress-energy tensor in the dynamical case, some energy condition should be assumed, and here we made a choice which facilitates the analysis and is satisfied in known models of gravitational collapse. 
As an example, we have computed the rate of evaporation explicitly in the Vaidya spacetime  and shown that the Schwarzschild spacetime can never be in equilibrium with the quantum matter field outside the horizon, if the quantum matter is in a state which is the vacuum in the asymptotic past.

The results obtained here should be regarded as a first step towards a more complete analysis of black hole evaporation in semiclassical gravity. A full solution of the semiclassical Einstein equations showing black hole evaporation is still lacking. This solution is available only in the two-dimensional case \cite{Ashtekar:2010hx,Ashtekar:2010qz}. The difficulties in controlling the state-dependent contributions in the expectation values of the stress-energy tensor prevent the generalization to the four-dimensional case. In this regard, a study similar to the one in \cite{Meda:2020smb} for cosmological spacetimes would be desirable.
	
\bigskip

{\subsection*{Acknowledgements}
We thank two anonymous referees for helpful comments on an earlier version of this paper.}

\bigskip

\appendix

\section{{Appendix}}
\label{app}

\subsection{Proof of \eq \eqref{eq:dM-TUV-J}}
\label{app:proof-DeltaM}

Using that $T^{UV} = A^{-2} T_{UV}$ and $T^{UU} = A^{-2} T_{VV}$ we can relate the current $J_1$ defined in \eq \eqref{eq:current-r} to the variation of the mass \eqref{eq:dM} computed along the line enclosed between $(V_0,U_1), (V_1,U_2) \in \delta \cH$ in the $(V,U)$ plane (see Figure \ref{fig:D}). In $(V,U)$ coordinates, the derivatives of the Misner-Sharp energy \eqref{eq:Hawking-mass} read
\begin{subnumcases}{}
	\pa_V m  = \frac{4\pi r^2}{A} \at {T_{UV}} \pa_V r - {T_{VV}} \pa_U r \ct \label{eq:Einstein-V}, \\
	\pa_U m  = \frac{4\pi r^2}{A} \at {T_{UV}} \pa_U r - {T_{UU}} \pa_V r \ct \label{eq:Einstein-U}.
\end{subnumcases}
Evaluating \eqs \eqref{eq:Einstein-V} and \eqref{eq:Einstein-U} on $\pa_V r \eqH 0$, we obtain that
\begin{equation}
\label{eq:M-J1}
	\begin{aligned}
		{\Delta M} = 4\pi \int_{\delta \cH} r^2 (-\pa_U r) \at \frac{{T_{VV}}}{A} \d V - \frac{{T_{VU}}}{A} \ct \d U = 4 \pi \int_{\delta \cH} A r^2 \at J_1^V \d U - J_1^U \d V \ct.
	\end{aligned}
\end{equation}
\Eq \eqref{eq:dM-TUV-J} can be obtained by applying the divergence theorem (Stokes' theorem) to the current $J_1$ on the domain $D\times \mathds{S}^2$. Using the spherical symmetry to integrate out the angular variables $(\varphi,\theta)$, we obtain that 
\begin{align*}
	- \int_{D} \at \nabla \cdot J_1 \ct \d \cV_D &= \int_{\delta \cH } (J_1^V Ar^2 \d U - J_1^U Ar^2 \d V ) + \int_{\rho_0} J_1^V Ar^2 \d U + \int_{\delta_0} J_1^U Ar^2 \d V - \int_{\gamma} J_1^V Ar^2 \d U. 
\end{align*}
With the choice of the initial conditions \eqref{eq:initial-condition}, both the integrals along $\rho_0$ and $\delta_0$ vanish. By substitution of \eq \eqref{eq:M-J1} at the place of the integral over $\delta \cH$, we get
\begin{equation*}
	- \int_{D} \at \nabla \cdot J_1 \ct \d \cV_D = \frac{\Delta M}{4\pi} - \int_{U_0}^{U_2} Ar^2 J_1^V \d U.
\end{equation*}
Thus, \eq \eqref{eq:dM-TUV-J} is obtained by employing the definition $J_1 = (J_r-J_K)/2$, where $J_r$ and $J_K$ are given in \eqs \eqref{eq:current-r} and \eqref{eq:current-k}, and by using that $\nabla \cdot J_1 = \frac{1}{2} \nabla \cdot J_r$, since $\nabla \cdot J_K=0$ everywhere.

\subsection{Proof of Theorem \ref{theo:evap}}
\label{app:proof-evap}

The proof consists in applying the divergence theorem (Stokes' theorem) on the domain $D\times \mathds{S}^2$ to a quantum current $\tilde{J}$ depending on $\expvalom{T_{\mu\nu}}$, which is a weighted version of $J_1$ given in \cref{eq:current-r,eq:current-k}. The weight is given in terms of a strictly positive function $f(V,U)$ which will be fixed later. Let us define
\[
	\tilde{J}_\mu \doteq  \xi^\nu\expvalom{T_{\nu\mu}}, \qquad \xi^\nu =  f(V,U)(\partial_V)^\nu 
\]
and the weighted variation of the mass
\begin{equation}
\label{eq:dfM}
	\Delta_h M \doteq \int_{\delta \cH} \hspace{0em} h \d m
\end{equation}
with respect to the function
\begin{equation}
\label{eq:h}
	h(V,U) \doteq \frac{f(V,U) A(V,U)}{-\pa_U r(V,U)}.
\end{equation}
The divergence of $\tilde{J}$ is related to the variation of the weighted mass \eqref{eq:dfM} by the following equation:
\begin{equation}
\label{eq:app-formula1}
	\Delta_h M = - 4\pi\int_{D} \nabla \cdot \tilde{J} \d \cV_D - 4\pi \int_{U_0}^{U_2} \frac{\expvalom{T_{UV}}}{A} f A r^2 \d U,
\end{equation}
where
\[
	\Delta_h M = 4 \pi \int_{\delta \cH} \hspace{-0em} r^2 \at \tilde{J}_V  \d V - \tilde{J}_U \d U \ct.
\]
\Eq \eqref{eq:app-formula1} can be obtained similarly to what already done in the Appendix \ref{app:proof-DeltaM} for the current $J_1$, namely by applying the divergence theorem (Stokes' theorem) to the weighted current $\tilde{J}_\mu$ on the domain $D\times \mathds{S}^2$, under the assumptions of Section \ref{sec:variation-mass} for the domain $D$ and imposing the initial conditions \eqref{eq:initial-condition} on $\expvalom{T_{\mu\nu}}$.

\medskip

Using the conservation equation $\nabla^\mu \expvalom{T_{\mu\nu}} = 0$, the relation \eqref{eq:trace-relation}, and the semiclassical equations $A^{-1} \expvalom{T_{UV}} = {R_\theta}^\theta/(8\pi)$, $A^{-1}\expvalom{T_{VV}} = -{R_V}^U/(8\pi)$,
\[
	\nabla \cdot \tilde{J} = \frac{1}{8\pi} \aq - \at -{R_V}^U \pa_U f + {R_\theta}^\theta \pa_V f \ct + f \at -{R_\theta}^\theta \frac{\pa_V A}{A} + 2 {R_\theta}^\theta \frac{\pa_V r}{r} + 8\pi \expvalom{{T_{\rho}}^{\rho}} \frac{\pa_V r}{r} \ct \cq.
\]
Here, $\expvalom{{T_{\rho}}^{\rho}}$ is a geometric quantity given in terms of the trace anomaly in \eq \eqref{eq:T-trace}. There we can isolate a positive contribution after computing explicitly the product
\[
	{C_{\alpha\beta}}^{\gamma\delta}{C_{\gamma\delta}}^{\alpha\beta} = \at R + \frac{12 \kappa}{r}\ct^2 = 4\at {R_U}^U + {R_\theta}^\theta + \frac{6 \kappa}{r}\ct^2
\]
and the difference 
\[
	{R_\mu}^\nu {R_\nu}^\mu - \frac{1}{3} R^2 = 2{R_V}^U {R_U}^V + \frac{2}{3} \at ({R_U}^U)^2 + \at {R_\theta}^\theta \ct^2 \ct - \frac{8}{3} {R_\theta}^\theta {R_U}^U.
\]
Hence, the anomaly can be rewritten as
\[
	\expvalom{{T_{\rho}}^{\rho}} = \lambda \at 4\at {R_U}^U + {R_\theta}^\theta + \frac{6 \kappa}{r}\ct^2 + \frac{2}{3} \at ({R_U}^U)^2 + \at {R_\theta}^\theta \ct^2 \ct + 2{R_V}^U {R_U}^V -\frac{8}{3}{R_U}^U{R_\theta}^\theta \ct,
\]
where the first two terms are manifestly positive.
Plugging this expression inside \eq \eqref{eq:app-formula1} yields
\begin{equation}
\label{eq:int-fdM}
	\begin{aligned}
		\Delta_h M  &= -  4\pi\lambda \int_{\text{D}} \at 4 \at {R_U}^U + {R_\theta}^\theta + \frac{6 \kappa}{r}\ct^2 + \at ({R_U}^U)^2 + \at {R_\theta}^\theta \ct^2 \ct \ct \frac{\pa_V r}{r} f \d \cV_\text{D} \\
					&- 4\pi \int_{\text{D}} \at - \frac{\pa_V f}{f} + \pa_V \log(A^{-1}r^2) - \lambda \frac{64\pi}{3}{R_U}^U \frac{\pa_V r}{r} \ct {\frac{{R_\theta}^\theta}{8\pi}} f \d \cV_\text{D} \\
					&- 4\pi\int_{\text{D}} \at {+} \frac{\pa_U f}{f} + {8\pi} \lambda \frac{2\pa_V r}{r} {{R_U}^V} \ct {\frac{{R_V}^U}{8\pi}} f \d \cV_\text{D} - 4\pi \int_{U_0}^{U_2} \frac{\expvalom{T_{UV}}r^2}{A} f(V_1,U) A\d U. 
	\end{aligned}
\end{equation}
Since $D$ is a normal domain, e.g., with respect to the $V$-axis for any $\cF(V,U) \in \cC^\infty(\cM)$, it holds that
\begin{equation}
	\int_{D} \cF(V,U)  \d \cV_D =  \int_{V_0}^{V_1} \d V \int_{U_0}^{U_{\cH}(V)} \cF(V,U) A r^2 \d U ,
\end{equation}
where $U_{\cH}(V)$ is the solution of $2r(V,U) - m(V,U) = 0$.

Our aim is to prove now that $\Delta_hM$ is strictly negative using \eq \eqref{eq:int-fdM}. To this aim, we shall isolate all the integrals in the right-hand side which give a negative contribution to $\Delta_hM$, while we tame the effects of the other choosing carefully the geometric function $f$. Actually, we want to find a function $f(V,U) > 0$ such that all these unwanted terms in \eq \eqref{eq:int-fdM} vanish. Let $\beta(V,U)$ be any fixed primitive function of 
\[
	\pa_U \beta(V,U) = \frac{2}{r} \pa_V r {{R_U}^V},
\]
The $U$-derivative of $f$ is fixed in such a way to cancel the volume integral whose integrand is proportional to ${R_V}^U$, namely it must be a solution of the equation 
\[
	\at \frac{\pa_U f}{f} + {8\pi} \lambda \frac{2\pa_V r}{r} {{R_U}^V} \ct = 0.
\]
Hence, we get
\[
	f(V,U) = f_0(V) \exp({-8\pi}\lambda \beta(V,U)),
\]
where $f_0(V)$ is an integration constant which can be chosen consistently with the hypothesis stated in the Theorem. Plugging this function $f(V,U)$ in the contributions of \eq \eqref{eq:int-fdM} yields
\begin{equation}
\label{eq:int-fdM-2}
	\begin{aligned}
		\Delta_h M =& -4\pi\int_{V_0}^{V_1}  {f_0}(V)\; ( \gamma_3(V) + \gamma_1(V)) \d V + 4\pi\int_{V_0}^{V_1}  {\pa_V f_0}(V)\; \gamma_2(V) \d V \\
		&- 4\pi \int_{U_0}^{U_2} \frac{\expvalom{T_{UV}}r^2}{A} f(V_1,U) A\d U,
	\end{aligned}
\end{equation}
where
\begin{align*}
	\gamma_1(V) &= \frac{1}{8\pi} \int_{U_0}^{U_{\cH}} {R_\theta}^\theta \at \pa_V \log(A^{-1}r^2) + 8\pi \lambda \at -\frac{8}{3}{R_U}^U \frac{\pa_V r}{r} + \pa_V \beta(V,U) \ct\ct
	\frac{f}{f_0} A r^2  \d \tilde{U}, \\
	\gamma_2(V) &= \frac{1}{8\pi} \int_{U_0}^{U_{\cH}} {R_\theta}^\theta \frac{f}{f_0} A r^2 \d \tilde{U}, 
\end{align*}
and $\gamma_3$ is  
\[
	\gamma_3(V)  =  \lambda\int_{U_0}^{U_{\cH}(V)} \at 4 \at {R_U}^U + {R_\theta}^\theta + \frac{6 \kappa}{r}\ct^2 + \frac{2}{3} \at ({R_U}^U)^2 + ({R_\theta}^\theta)^2 \ct \ct \frac{\pa_V r}{r} \frac{f}{f_0} A r^2  \d \tilde{U}.
\]

To prove that $\Delta_hM$ is strictly negative, the three contributions given by the three integrals on the right hand side of \eq \eqref{eq:int-fdM-2} are analyzed separately. Since $\partial_V r>0$ outside the horizon, $\gamma_3$ can be controlled as follows: 
\[
	\gamma_3(V) \geq \frac{144}{13} \lambda\int_{U_0}^{U_{\cH}(V)} \kappa^2  \frac{\pa_V r}{r} \frac{f}{f_0} A   \d U.
\]
From the behaviour on the apparent horizon of the expansion parameters of the ingoing and outgoing radial null geodesics given in \eq \eqref{eq:FOTH}, and according to the definition of $\kappa$ given in \eqref{eq:kodama-box}, $\kappa$ is strictly positive on the apparent horizon, and by continuity it stays positive also near the horizon. Thus, $\gamma_3(V)$ is strictly positive for $V \in [V_0,V_1]$.

Moreover, the initial conditions given on the hypersurface $\rho_0$ which is part of $\pa D$ imply that ${R_\theta}^\theta = A^{-1} \expvalom{T_{UV}} = 0 $ on $\rho_0$, and hence $\gamma_1(V_0) = \gamma_2(V_0) = 0$. Then, $\gamma_1+\gamma_3$ is strictly positive for $V=V_0$, and by continuity it stays strictly positive also for $V$ near $V_0$. Therefore, we may find a constant $\delta>0$ such that $f_0(\gamma_3+\gamma_1)$ is strictly positive on $[V_0,V_0+\delta]$. If $k$ in $f_0$ is sufficiently large, the integral of $f_0(\gamma_3+\gamma_1)$ over $[V_0,V_1]$ is dominated by the contribution on $[V_0,V_0+\delta]$. Hence, the first contribution in the right-hand side of \eq \eqref{eq:int-fdM-2} containing $\gamma_3+\gamma_1$ is strictly negative for that choice of $k$. 

Furthermore, the term containing $\gamma_2$ in $\Delta_hM$ in \eq \eqref{eq:int-fdM-2} is negative or null, because $\pa_V f_0 < 0$ on $[V_0,V_1]$, ${R_\theta}^\theta = A^{-1} \expvalom{T_{UV}} = 0$, and $\gamma_2(V) \geq 0$ for all $V\in [V_0,V_1]$, according to the hypothesis stated in \eq \eqref{eq:AEC}. 

Finally, the condition \eqref{eq:AEC} also implies that the last integral appearing in $\Delta_h M$ in \eq \eqref{eq:int-fdM-2}, which is computed for $V\in V_1$ and supported in $[U_0,U_2]$, gives a negative (or null) contribution to $\Delta_h M$.

Taking into account all this and with this choice of $f(V,U)$, $h(V,U)$ given in \eq \eqref{eq:h} is also positive and smooth. Hence, it is bounded from below in $\delta \cH$, so $0 > \Delta_h M \geq C{\Delta M}$, where $C > 0$, and the proof of the Theorem holds.

\newcommand{\SortNoop}[1]{}

\end{document}